\documentclass{article}
\usepackage[a4paper]{geometry}
\usepackage{authblk}
\usepackage{amsmath}
\usepackage{amssymb}
\usepackage{latexsym}
\usepackage{amsthm} 
\usepackage{array}
\usepackage{paralist}
\usepackage{algorithm, algorithmic}
\usepackage{graphicx} 
\usepackage{subfig}
\usepackage{multirow} 
\usepackage{comment}
\usepackage{tikz}
\usetikzlibrary{shapes,arrows,positioning,chains,decorations.pathreplacing,calc}
\newcommand{\interestahead}{*}
\newcommand{\interestdhead}{square}

\newtheorem{definition}{Definition}

\newtheorem{proposition}{Proposition}

\usepackage{url}
\urldef{\mailsa}\path|jason.crampton@rhul.ac.uk|
\urldef{\mailsb}\path|james.sellwood.2010@live.rhul.ac.uk|

\newcommand{\set}[1]{\ensuremath{\left\{#1\right\}}} 

\newcommand{\comp}{\mathbin{;}}

\newcommand{\audita}[1]{\ensuremath{#1^\oplus}}

\newcommand{\auditd}[1]{\ensuremath{#1^\ominus}}

\newcommand{\interesta}{\ensuremath{i^\oplus}}
\newcommand{\interestb}{\ensuremath{i^\ominus}}
\newcommand{\setinterest}{\ensuremath{\set{i^\oplus,i^\ominus}}}

\newcommand{\pa}{\mathit{PA}}

\renewcommand{\mp}{\mathit{MP}}

\newcommand{\pd}{\mathit{PD}}

\begin{document}

\title{Caching and Auditing in the RPPM Model}

\author{Jason Crampton}
\author{James Sellwood}
\affil{Information Security Group\\ Royal Holloway, University of London}

\maketitle
\begin{abstract}
    Crampton and Sellwood recently introduced a variant of relationship-based access control based on the concepts of \underline{r}elationships, \underline{p}aths and \underline{p}rincipal \underline{m}atching, to which we will refer as the RPPM model.
    In this paper, we show that the RPPM model can be extended to provide support for caching of authorization decisions and enforcement of separation of duty policies.
    We show that these extensions are natural and powerful.
    Indeed, caching provides far greater advantages in RPPM than it does in most other access control models and we are able to support a wide range of separation of duty policies.
\end{abstract}

\section{Introduction}
Whilst the majority of computer systems employ some form of role-based access control model, social networking sites have made use of the relationships between individuals as a means of determining access to resources.
Recent work on relationship-based access control models has attempted to further develop this concept but has frequently remained focused on the relationships that exist between individuals~\cite{CarminatiFP09,Fong11}.
Crampton and Sellwood define a more general model for access control utilising relationships between entities, where those entities can represent any physical or logical component of a system~\cite{CramptonS14}.
These entities and their (inter-)relationships are described by a multigraph, called the \emph{system graph}.
Authorization requests in the RPPM model are processed by first determining a list of matching principals.
This list of principals is identified using principal-matching rules and the system graph.
Once a list of matched principals is determined, the specific action in the request is authorized or denied based on authorization rules defined for those principals and the object.

The RPPM model provides the necessary foundations for general purpose relationship-based access control systems, but there are a number of simple enhancements which would greatly increase its utility and efficiency.
The evaluation of path conditions can be complex in system graphs containing many nodes of high degree.
Support for caching of previously matched principals would significantly reduce the processing necessary during the evaluation of an authorization request.
The introduction of caching support is, therefore, our first enhancement.

Our second enhancement adds support for request evaluation audit records to be kept, such that future authorization requests may be evaluated both on the current relationships within the system graph but also using historical information about past actions by subjects.
Such mechanisms allow us to support constraints such as separation of duty and Chinese Wall policies and lay a foundation for future work on workflow authorization using the model.

The rest of this paper is arranged with background information on the RPPM model provided in Section~\ref{sec:preliminaries} and then the two enhancements described individually in Sections~\ref{sec:caching} (caching) and~\ref{sec:audit} (audit records).
We discuss related work in Section~\ref{sec:relatedwork} and draw conclusions of our contributions and identify future work in Section~\ref{sec:conclusion}.

\section{The RPPM Model}\label{sec:preliminaries}
The RPPM model, described in detail in~\cite{CramptonS14}, employs a system graph to capture the entities of a system and their (inter-)relationships. The entities (physical or logical system components) are nodes within the system graph whilst the relationships are labelled edges. The system graph's `shape' is constrained by a system model, which identifies the types of entities and relationship which are supported. It does so by defining a permissible relationship graph whose nodes are the possible types of entities in the system graph and whose labelled edges indicate the relationships which \emph{may} exist in the system graph between entities of the connected types.

\begin{definition}
    A \emph{system model} comprises a set of types $T$, a set of relationship labels $R$, a set of \emph{symmetric} relationship labels $S \subseteq R$ and a \emph{permissible relationship graph} $G_{\textrm{PR}} = (V_{\textrm{PR}},E_{\textrm{PR}})$, where $V_{\textrm{PR}} = T$ and $E_{\textrm{PR}} \subseteq T \times T \times R$.
\end{definition}

\begin{definition}
    Given a system model $(T,R,S,G_{\textrm{PR}})$, a \emph{system instance} is defined by a \emph{system graph} $G = (V,E)$ where $V$ is the set of entities and $E \subseteq V \times V \times R$.
    Making use of a mapping function $\tau : V \rightarrow T$ which maps an entity to its type, we say $G$ is \emph{well-formed} if for each entity $v$ in $V$, $\tau(v) \in T$, and for every edge $(v,v',r) \in E$, $(\tau(v),\tau(v'),r) \in E_{\textrm{PR}}$.
\end{definition}

Within the RPPM model, authorization requests have the form $q = (s, o, a)$, where a subject $s$ requests authorization to perform action $a$ on target object $o$.
The authorization policy is abstracted away from subjects by the use of security principals.
These principals are matched to requests through the satisfaction of path conditions using edges in the system graph, where a path condition $\pi$ represents a sequence of relations with specific labels from the set $R$.

\begin{definition}\label{def:path-condition}
   Given a set of relationships $R$, we define a \emph{path condition} recursively:
    \begin{itemize}
        \item $\diamond$ is a path condition;
        \item $r$ is a path condition, for all $r \in R$;
        \item if $\pi$ and $\pi'$ are path conditions, then $\pi \comp \pi'$, $\pi^+$ and $\overline{\pi}$ are path conditions.
    \end{itemize}
   A path condition of the form $r$ or $\overline{r}$, where $r \in R$, is said to be an \emph{edge condition}.
\end{definition}

Informally, $\pi \comp \pi'$ represents the concatenation of two path conditions; $\pi^+$ represents one or more occurrences, in sequence, of $\pi$; and $\overline{\pi}$ represents $\pi$ reversed; $\diamond$ defines an ``empty'' path condition.

\begin{definition}
   Given a set of relationships $R$, we define a \emph{simple path condition} recursively:
    \begin{itemize}
        \item $\diamond$, $r$ and $\overline{r}$, where $r \in R$, are simple path conditions;
        \item if $\pi \ne \diamond$ and $\pi' \ne \diamond$ are simple path conditions, then $\pi \comp \pi'$ and $\pi^+$ are simple path conditions.
    \end{itemize}
\end{definition}

A path condition can describe highly complex and variable-length paths within the system graph.
However, Crampton and Sellwood proved that every path condition can be reduced to an equivalent simple path condition~\cite[\S{2.2}]{CramptonS14}, thereby simplifying the design of the principal-matching algorithm.

\begin{definition}\label{def:path-condition-semantics}
    Given a system graph $G = (V,E)$ and $u,v \in V$, we write $G,u,v \models \pi$ to denote that  $G$, $u$ and $v$ \emph{satisfy path condition} $\pi$.
    Formally, for all $G,u,v,\pi,\pi'$:
    \begin{itemize}
        \item $G,u,v \models \diamond$ iff $v = u$;
        \item $G,u,v \models r$ iff $(u,v,r) \in E$;
        \item $G,u,v \models \pi \comp \pi'$ iff there exists $w \in V$ such that $G,u,w \models \pi$ and $G,w,v \models \pi'$;
        \item $G,u,v \models \pi^+$ iff $G,u,v \models \pi$ or $G,u,v \models \pi \comp \pi^+$;
        \item $G,u,v \models \overline{\pi}$ iff $G,v,u \models \pi$.
    \end{itemize}
\end{definition}

\begin{definition}
    Let $P$ be a set of authorization principals.
    A \emph{principal-matching rule} is a pair $(\pi,p)$, where $\pi$ is a path condition and $p \in P$ is the associated principal.
    A list of principal-matching rules is a \emph{principal-matching policy}.
\end{definition}

In the context of a principal-matching rule, a path condition is called the \emph{principal-matching condition}.

The request and system graph are evaluated against the principal-matching policy utilising a \emph{principal-matching strategy} (PMS) to determine the list of matched principals for the request.
The PMS specifies how principal-matching rules should be evaluated, for example whether the first matching principal applies (in the case of the \textsf{FirstMatch} PMS) or whether all matching principals apply (\textsf{AllMatch}).
A \emph{default} principal-matching rule $(\top, p')$ may, optionally, be employed as the last rule in the policy and will automatically result in its principal $p'$ being matched whenever the rule is evaluated.

A system graph $G$, two nodes $u$ and $v$ in $G$, a principal-matching policy $\rho$, and a principal-matching strategy $\sigma$ determines a list of principals $\mp$ associated with the pair $(u,v)$.
We evaluate each principal-matching rule $(\pi,p)$ in turn and add $p$ to the list of matched principals if and only if $G,u,v \models \pi$.
We then apply the principal-matching strategy to the list of matched principals to obtain $\mp$.
(Obviously, optimizations are possible for certain principal-matching strategies.)
We write $G,u,v \xrightarrow{\rho,\sigma} \mp$ to denote this computation.

Once determined, the list of matched principals is used to identify relevant authorization rules in the authorization policy.

\begin{definition}
    An \emph{authorization rule} has the form $(p,o,a,b)$, where $p$ is a principal, $o$ is an object, $a$ is an action and $b \in \set{0,1}$, where $b = 0$ denies the action and $b = 1$ grants the action.
    In order to ease authorization policy specification we allow for the use of $\star$ instead of $o$, to represent all objects, or instead of $a$, to represent all actions.
    These global authorization rules, therefore, have the form $(p,\star,a,b)$, $(p,o,\star,b)$ or $(p, \star, \star, b)$.
    An \emph{authorization policy} is a list of authorization rules.
\end{definition}

The matching of principals to authorization rules yields a list of authorization decisions, which is reduced to be single decision using a \emph{conflict resolution strategy} (CRS).
The CRS is used in much the same way as a rule-combining or policy-combining algorithm is used in XACML.
It may specify that particular outcomes are prioritised, such as (\textsf{AllowOverride} or \textsf{DenyOverride}), or that the first conclusive decision should be used (\textsf{FirstMatch}).

To summarise, given a request $(s,o,a)$, where $s$ and $o$ are nodes in the system graph and $a$ is an action, we first compute the list of matched principals $G,s,o \xrightarrow{\rho,\sigma} \mp$.
We then use $\mp$ and the authorization policy to determine which actions are granted and denied for those principals and apply the CRS to determine a final decision.
In this paper, we assume the use of the \textsf{AllMatch} PMS and \textsf{DenyOverride} CRS throughout.

\section{Caching}\label{sec:caching}
The most complex part of evaluating an authorization request in the RPPM model is the principal matching stage~\cite[\S{3}]{CramptonS14}.
This process attempts to satisfy path conditions within principal-matching rules using paths between the subject and the object of the request.
It is important to note that the requested action is immaterial during this processing stage (only becoming relevant when the authorization rules are considered).
The list of matched principals for a subject-object pair remains static until a change is made to the system graph or certain associated policy components.
Even then, not all of the possible changes would impact the matched principals between a particular subject and object.

We introduce the concept of \emph{caching edges} and make use of the relative stability of matched principals in order to reduce the processing required for future authorization requests.
We first redefine the system graph to support these new edges.
In particular, when we evaluate a request $(s,o,a)$ that results in a list of matched principals $\mp$, we add an edge $(s,o,\mp)$ to the system graph, directed from $s$ to $o$ and labelled with $\mp$.

%

Informally, a caching edge $(s,o,\mp)$ directly links $s$ to $o$ and identifies the matching principals $\mp$ relevant to requests of the form $(s,o,a)$.
The processing of subsequent authorization requests can skip the principal matching stage and use $\mp$ in conjunction with the authorization rules to evaluate a request of the form $(s,o,a)$.

To illustrate, consider the simple system graph $G_1$, shown in Figure~\ref{img:system-graph-fragment1:a}, and the following principal-matching and authorization policies
\begin{align*}
 \rho &= [(r_1, p_1),(r_2, p_2),(r_3, p_3),(r_1 \comp r_3, p_4),(r_2 \comp r_3, p_5)] \\
 \pa &= [(p_5,\star,a_1,1),(p_5,\star,a_2,0)].
\end{align*}

If an authorization request $q_1 = (v_2,v_4,a_1)$ is made, then \mbox{$G_1,v_2,v_4 \xrightarrow{\rho,\sigma} [p_5]$}, because the only principal-matching condition from the policy which can be satisfied between $v_2$ and $v_4$ in $G_1$ is $ r_2 \comp r_3$.
Then the authorization rule $(p_5,\star,a_1,1)$ applies and the set of possible decisions $\pd = \set{1}$; thus the request is authorized.
At this stage we may add a caching edge $(v_2,v_4,[p_5])$ to produce the system graph shown in Figure~\ref{img:system-graph-fragment1:b}.
We use the convention that caching edges have a diamond-shaped arrow head.

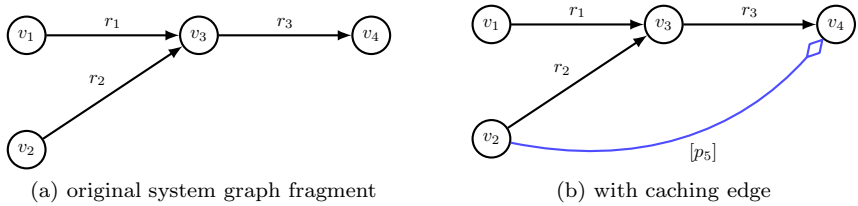
\begin{figure}[!ht]\centering
    \subfloat[original system graph fragment]{
          \begin{tikzpicture}
              [node distance=1cm and 1.75cm,
            caching/.style={color=blue!70,>=open diamond}, 
            audit-a/.style={color=teal, densely dashed}, 
            audit-d/.style={color=red, densely dotted}, 
            interest-a/.style={color=purple,>=\interestahead}, 
            interest-b/.style={color=brown,>=\interestdhead}, 
              every circle node/.style={draw,minimum width=20pt},thick,
              every node/.append style={scale=0.7, transform shape}]
              \begin{scope}[>=latex] 
            	  \node[circle] (v1) {$v_1$};
            	  \node[circle,below=of v1] (v2) {$v_2$};
            	  \node[circle,right=of v1] (v3) {$v_3$};
            	  \node[circle,right=of v3] (v4) {$v_4$};
            	  \draw[thick,->] (v1) to node[auto] {\textsf{$r_1$}} (v3);
            	  \draw[thick,->] (v2) to node[auto] {\textsf{$r_2$}} (v3);
            	  \draw[thick,->] (v3) to node[auto] {\textsf{$r_3$}} (v4);
              \end{scope}
          \end{tikzpicture}
        \label{img:system-graph-fragment1:a}
    }\qquad
     \subfloat[with caching edge]{
          \begin{tikzpicture}
              [node distance=1cm and 1.75cm,
            caching/.style={color=blue!70,>=open diamond}, 
            audit-a/.style={color=teal, densely dashed}, 
            audit-d/.style={color=red, densely dotted}, 
            interest-a/.style={color=purple,>=\interestahead}, 
            interest-b/.style={color=brown,>=\interestdhead}, 
              every circle node/.style={draw,minimum width=20pt},thick,
              every node/.append style={scale=0.7, transform shape}]
              \begin{scope}[>=latex] 
            	  \node[circle] (v1) {$v_1$};
            	  \node[circle,below=of v1] (v2) {$v_2$};
            	  \node[circle,right=of v1] (v3) {$v_3$};
            	  \node[circle,right=of v3] (v4) {$v_4$};
            	  \draw[thick,->] (v1) to node[auto] {\textsf{$r_1$}} (v3);
            	  \draw[thick,->] (v2) to node[auto] {\textsf{$r_2$}} (v3);
            	  \draw[thick,->] (v3) to node[auto] {\textsf{$r_3$}} (v4);
            	  \draw[caching,thick,->] (v2) to [bend right] node[swap,auto,color=black] {\textsf{$[p_5]$}} (v4);
              \end{scope}
          \end{tikzpicture}
        \label{img:system-graph-fragment1:b}
    }
    \caption{Adding a caching edge}\label{img:system-graph-fragment1}
\end{figure}

If an authorization request $q_2 = (v_2,v_4,a_2)$ is subsequently made, the caching edge $(v_2,v_4,[p_5])$ allows us to evaluate the request without re-evaluating the principal-matching policy; the authorization rule $(p_5,\star,a_2,0)$ subsequently results in request $q_2$ being denied ($\pd = \set{0}$).

To consider the scale of the potential benefit of caching edges, we review the experimental data reported by Crampton and Sellwood for numbers of nodes visited ($n$) and edges considered ($e$) during sample request evaluations (see Table~\ref{tbl:implementation_metrics} and~\cite[\S{3.3}]{CramptonS14}).
With support for caching edges, if the subject-object pairs participating in any of these requests were to be involved in subsequent requests the processing would instead be limited to locating the appropriate caching edge.
It should be clear that when considering requests which may require upwards of 50 edge evaluations (in a small example system graph), replacing this with a single caching edge lookup could dramatically improve evaluation performance.
%

\begin{table}[!ht]\centering
  \caption{Experiment results from~\cite[\S{3.3}]{CramptonS14}}\label{tbl:implementation_metrics}
  {\renewcommand{\arraystretch}{1.25}
  \begin{tabular}{|c|c|r|r|c|}
    \hline
        \bf Path condition & \bf Request & $n$ & $e$ & \bf Match Found\\
    \hline
    \hline
        $\pi_1$ & $q_1$ & 5 & 19 & Yes\\
        $\pi_1$ & $q_2$ & 7 & 24 & Yes\\
        $\pi_2$  & $q_3$ & 4 & 15 & Yes\\
        $\pi_3$ & $q_4$ & 17 & 58 & Yes\\
        $\pi_3$ & $q_5$ & 7 & 24 & No\\
    \hline
  \end{tabular}}
\end{table}

In the worst case, the number of caching edges directed out of a node is $O(|V|)$, where $V$ is the set of nodes in the system graph.
However, there are strategies that can be used to both prevent the system graph realizing the worst case and to reduce the impact of large numbers of caching edges.
To maintain an acceptable number of caching edges, we could, for example, use some form of cache purging.
We can also distinguish between relationship edges and caching edges using some flag on the edges and index the caching edges to dramatically decrease the time taken to search the set of caching edges.
Employing these techniques should enable the benefits of caching edges to be realised without incurring unacceptable costs during identification of the relevant caching edge.
Further experimental work is required to determine how best to make use of caching edges.

\subsection{Preemptive Caching}\label{sec:caching:preemptive}
Any optimisation provided by the caching of matched principals relies upon the existence of a caching edge in order to reduce the authorization request processing; the first request between a subject and object must, therefore, be processed normally in order to determine the list of matched principals which will label the caching edge.
If this initial evaluation were only performed when an authorization request were submitted, then the benefit of caching edges would be limited to repeated subject-object interactions alone.

However, many authorization systems will experience periods of time when no authorization requests are being evaluated.
The nature of many computing tasks is such that authorization is required sporadically amongst longer periods of computation by clients of the authorization system and idle time for the authorization system itself.
These periods of reduced load on the authorization system can be employed for the purpose of \emph{preemptive caching}.

Thus for pairs of nodes $(u,v)$ in the system graph, we may compute $G,u,v \xrightarrow{\rho,\sigma} \mp$ and insert a caching edge $(u,v,\mp)$.
The fact that a request's action is not employed during the principal matching process means that to perform this further optimization an authorization system solely requires a subject and object between whom the matched principals are to be identified.
There are numerous potential strategies for determining which subject-object pairs should be considered for preemptive caching.
Here we describe two simple and natural strategies.

\begin{description}
    \item[Subject-focused.]
    Subject-focused preemptive caching assumes that subjects who have recently made authorization requests are \emph{active} and so will likely make further requests.
    The authorization system, therefore, prioritises determining the list of matched principals between the most recently active subjects and a set of target objects.
    The set of target objects could be selected at random or may be systematically chosen using an appropriate mechanism for the system defined in the system graph.
    This might involve the target objects being \emph{popular}, \emph{significant} or those whose access may be particularly \emph{time-sensitive}.
    We envisage that the interpretation of these concepts may be system specific, as may their relative worth.

    As preemptive caching builds the number of caching edges within the system graph the number of subjects and objects under consideration could be expanded to provide greater coverage of the potential future requests.

   \item[Object-focused.]
   In certain applications, there will be resources that will be used by most users, such as certain database tables.
   Thus, it may make sense to construct caching edges for all active users for certain resources.
\end{description}

No matter the strategy, preemptive caching makes use of available processing time in order to perform the most complex part of authorization request evaluation: principal matching.
Any requests that are made utilising a subject-object pair which have already been evaluated by preemptive caching will be able to make use of the caching edge already established, even if that request were the first received for that pair.
Once determined, caching edges resulting from preemptive caching are no different from those established as a result of request evaluation.

\subsection{Cache Management}\label{sec:caching:management}
A change to any of the following components of the model could modify the list of matched principals for a subject and object:
\begin{itemize}
    \item the system graph;
    \item the principal-matching policy;
    \item the principal-matching strategy.
\end{itemize}
Such changes, therefore, may affect the correctness of caching edges.
(The obvious exception is a change to the system graph resulting from the addition or deletion of a caching edge.)
The most crude management technique for handling such changes involves removing all caching edges from the system graph whenever one of the above changes occurs.

In certain specific scenarios it may be possible for a system to identify a scope of impact for a particular change and thus apply a more refined management technique.
For example, if a change to the principal-matching policy removes all rules which are used to match a certain principal (and nothing more), then it would be sufficient for only caching edges labelled with a list including that principal to be purged.
Whilst such a refinement may further optimise the operations performed by the authorization system, its applicability will depend upon the configuration of the authorization system in its entirety.

We have already noted that it may make sense to purge the cache in order to limit the number of caching edges in the system graph.
Again, there are several possible purging strategies.
One would be simply to set a maximum threshold for the number of caching edges in the system graph.
A second, perhaps more useful, strategy would be to set a maximum threshold for the out-degree (measured in terms of caching edges) for any node in the graph.
We may also ``retire'' caching edges: any edge that hasn't been used as part of a request evaluation for some time period will be purged.
And we could employ mixed strategies, which might depend on the application and the nature of the system graph.

\section{Audit Records}\label{sec:audit}
Currently, the RPPM model's authorization request processing is ``memoryless'' with respect to previous requests and their respective outcomes.
Various scenarios and security policy principles make use of historical data.
Reputation systems and history-based access control (HBAC) systems~\cite{AbadiF03,EdjlaliAC99,KrukowNS08}, for example, rely on knowledge of previous interactions and requests in order to correctly make authorization decisions.
The Chinese Wall~\cite{BrewerN89} and non-static separation of duty principles~\cite{GligorGF98,SimonZ97} also rely on knowledge of previous actions to enforce their constraints.

We introduce the concept of \emph{audit edges}, through which we track the outcomes of authorization requests for subsequent use in policy evaluation.
Audit edges come in two flavours: those which directly record the decision of a previous authorization request (authorized and denied \emph{decision audit edges}) and those which, more generally, record an entity's interest in other entities based on its authorized requests (active and blocked \emph{interest audit edges}).
It should be noted that whilst we make direct use of audit edges for policy evaluation, they also have value in a system purely as an audit record.
We extend the set of relationships and further redefine the system graph to support these new edges.
Specifically, in the case of decision audit edges:
\begin{itemize}
    \item for each action $a$, we define two relationships $\audita{a}$ and $\auditd{a}$ and include the sets $\set{\audita{a} : a \in A}$ and $\set{\auditd{a} : a \in A}$ in the set of relationships;
    \item if the decision for request $(u,o,a)$ is allow, then we add the edge $(u,o,\audita{a})$ into the system graph;
    \item if the decision for request $(u,o,a)$ is deny, then we add the edge $(u,o,\auditd{a})$ into the system graph.
\end{itemize}

Both authorized and denied decision audit edges are inserted, automatically, into the system graph after request evaluation completes.
If such an edge does not already exist, a decision audit edge is added between the subject and object of the evaluated request, indicating its result.

The addition of interest audit edges also occurs automatically after request evaluation completes.
For such edges, the subject is the source node of the interest edge (as for decision audit edges); however, the destination node may not be the object of the request.
Interest audit edges are discussed in more detail in Section~\ref{sec:audit:chinesewall}.

\subsection{Enforcing Separation of Duty}\label{sec:audit:sod}
Separation of duty requires that certain combinations of actions are performed by a number of distinct individuals so as to reduce the likelihood of abuse of a system.
In its simplest form, separation of duty constraints require two individuals to each perform one of a pair of distinct actions so that a single individual cannot abuse the system.
A common application environment for such constraints is that of a finance system, where, for example, the individual authorized to add new suppliers should not be the same individual who is authorized to approve the payment of invoices to suppliers.
If a single individual were able to perform both of these actions they could set themselves up as a supplier within the finance system and then approve for payment any invoices they submitted as that supplier.
We define a mechanism here through which $n$ individuals can be required to perform $n$ actions on an object.
Before doing so, we explain a simplified version of the mechanism for the case $n = 3$.

Let us consider the system graph $G_2$ (see Figure~\ref{img:system-graph-fragment2:a}), the principal-matching policy $\rho =  [(r, p)]$ and the authorization policy $\pa = [(p,o,\star,1)]$.
With these policies and without audit edges, if individual $u_1$ makes the request $q_1 = (u_1,o,a_1)$ this will be authorized by matching principal $p$, as will subsequent requests $q_2 = (u_1,o,a_2)$ and $q_3 = (u_1,o,a_3)$.
A similar result would have occurred if these requests had been submitted with $u_2$ or $u_3$ as the subject.

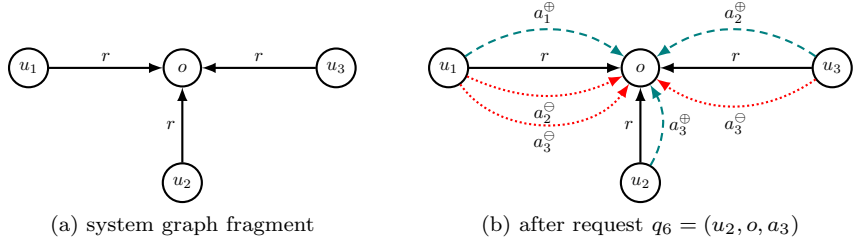
\begin{figure}[!ht]\centering
    \subfloat[system graph fragment]{
        \begin{tikzpicture}
            [node distance=1cm and 1.5cm,
            caching/.style={color=blue!70,>=open diamond}, 
            audit-a/.style={color=teal, densely dashed}, 
            audit-d/.style={color=red, densely dotted}, 
            interest-a/.style={color=purple,>=\interestahead}, 
            interest-b/.style={color=brown,>=\interestdhead}, 
              every circle node/.style={draw,minimum width=20pt},thick,
              every node/.append style={scale=0.7, transform shape}]
              \begin{scope}[>=latex] 
            	  \node[circle] (u1) {$u_1$};
            	  \node[circle,right=of u1] (o) {$o$};
            	  \node[circle,below=of o] (u2) {$u_2$};
            	  \node[circle,right=of o] (u3) {$u_3$};
            	  \draw[thick,->] (u1) to node[auto] {\textsf{$r$}} (o);
            	  \draw[thick,->] (u2) to node[auto] {\textsf{$r$}} (o);
            	  \draw[thick,->] (u3) to node[swap,auto] {\textsf{$r$}} (o);
              \end{scope}
        \end{tikzpicture}
        \label{img:system-graph-fragment2:a}
    }\qquad
     \subfloat[after request $q_6 = (u_2,o,a_3)$]{
        \begin{tikzpicture}
            [node distance=1cm and 2cm,
            caching/.style={color=blue!70,>=open diamond}, 
            audit-a/.style={color=teal, densely dashed}, 
            audit-d/.style={color=red, densely dotted}, 
            interest-a/.style={color=purple,>=\interestahead}, 
            interest-b/.style={color=brown,>=\interestdhead}, 
              every circle node/.style={draw,minimum width=20pt},thick,
              every node/.append style={scale=0.7, transform shape}]
              \begin{scope}[>=latex] 
            	  \node[circle] (u1) {$u_1$};
            	  \node[circle,right=of u1] (o) {$o$};
            	  \node[circle,below=of o] (u2) {$u_2$};
            	  \node[circle,right=of o] (u3) {$u_3$};
            	  \draw[thick,->] (u1) to node[auto] {\textsf{$r$}} (o);
            	  \draw[thick,->] (u2) to node[auto] {\textsf{$r$}} (o);
            	  \draw[thick,->] (u3) to node[swap,auto] {\textsf{$r$}} (o);
                  \draw[audit-a,thick,->] (u1) to [bend left=35] node[auto,color=black] {\textsf{\audita{a_1}}} (o);
                  \draw[audit-d,thick,->] (u1) to [bend right=25] node[swap,auto,color=black] {\textsf{\auditd{a_2}}} (o);
                  \draw[audit-d,thick,->] (u1) to [bend right=55] node[swap,auto,color=black] {\textsf{\auditd{a_3}}} (o);
                  \draw[audit-a,thick,->] (u2) to [bend right] node[swap,auto,color=black] {\textsf{\audita{a_3}}} (o);
                  \draw[audit-a,thick,->] (u3) to [bend right=35] node[swap,auto,color=black] {\textsf{\audita{a_2}}} (o);
                  \draw[audit-d,thick,->] (u3) to [bend left=35] node[auto,color=black] {\textsf{\auditd{a_3}}} (o);
              \end{scope}
        \end{tikzpicture}
        \label{img:system-graph-fragment2:d}
    }
    \caption{Adding decision audit edges}\label{img:system-graph-fragment2}
\end{figure}

A basic implementation of separation of duty can be employed by introducing a new principal $p_{seen}$ which matches if a user has performed any action on the object.
We change the principal-matching and authorization policies to
\[
  [(\audita{a_1},p_{seen}),(\audita{a_2},p_{seen}),(\audita{a_3},p_{seen}),(r,p)]\quad\text{and}\quad [(p_{seen},o,\star,0),(p,o,\star,1)]
\]
respectively\footnote{We assume the use of the \textsf{AllMatch} PMS and the \textsf{DenyOverride} CRS, but we could equally employ the \textsf{FirstMatch} PMS with any CRS as long as we ensure that the constraint principal-matching rules are added before any existing rules.}.
Using this combination of policies means that any user who has performed an action on object $o$ is prevented from performing another action as all actions are denied to the principal $p_{seen}$.

Whilst this basic implementation fulfils the requirement that no user may perform more than one action on the object, we may wish to specify more refined separation of duty policies within the system.
The basic implementation has several limitations which RPPM's policies are flexible and powerful enough to resolve.
Specifically, all actions within the system are included in the separation of duty constraint due to the use of an authorization rule for all actions $(p_{seen},o,\star,0)$.
Additionally, having performed an action on $o$ a user is unable to repeat the action performed, as well as being unable to perform any other action.

If we wish to enforce a more flexible separation of duty constraint on a subset of actions $\set{a_1, a_2, a_3} \subseteq A$ such that distinct individuals are required to perform each action, we can modify the principal-matching policy to $\rho = [(\audita{a_1},p_1),(\audita{a_2},p_2),(\audita{a_3},p_3),(r,p)]$ and the authorization policy to:
\begin{align*}
    \pa = [&(p_1,o,a_2,0),(p_1,o,a_3,0),
                \;(p_2,o,a_1,0),(p_2,o,a_3,0),\\
                &(p_3,o,a_1,0),(p_3,o,a_2,0),
                \;(p,o,\star,1)]
\end{align*}

\begin{description}
    \item[The first action.]
    Revisiting our example for $G_2$, an initial request $q_1 = (u_1,o,a_1)$ will, once again, be authorized (with $\mp = [p]$) but will, this time, result in the addition of an authorized decision audit edge $(u_1,o,\audita{a_1})$. 
    If $u_1$ then makes a request $q_2 = (u_1,o,a_2)$ this will be denied as $\mp = [p_1,p]$ and the authorization rule $(p_1,o,a_2,0)$ indicates a deny which overrides the authorization from the rule $(p,o,\star,1)$.
    Similarly if $u_1$ makes a request $q_3 = (u_1,o,a_3)$ this will be denied as once again $\mp = [p_1,p]$ and the deny authorization rule $(p_1,o,a_3,0)$ overrides $(p,o,\star,1)$.
    These two denied requests would result in denied decision audit edges $(u_1,o,\auditd{a_2})$ and $(u_1,o,\auditd{a_3})$. 

    \item[The second action.]
    However, if $u_3$ makes the request $q_4 = (u_3,o,a_2)$ this will be authorized with $\mp = [p]$ and use of the authorization rule $(p,o,\star,1)$; the authorized decision audit edge $(u_3,o,\audita{a_2})$ results.
    If $u_3$ attempts to then make request $q_5 = (u_3,o,a_3)$ this will be denied in the same manner that request $q_3$ was, with the subsequent addition of a denied decision audit edge $(u_3,o,\auditd{a_3})$.

    \item[The last action.]
    As $a_1$ was performed by $u_1$ and $a_2$ was performed by $u_3$ it remains, for successful operation, for $u_2$ to make request $q_6 = (u_2,o,a_3)$.
    This request will be authorized with $\mp = [p]$ and the use of the authorization rule $(p,o,\star,1)$, resulting in the authorized decision audit edge $(u_2,o,\audita{a_3})$.
    The system graph that results after all of these requests have been made is as shown in Figure~\ref{img:system-graph-fragment2:d}.
\end{description}

More generally, suppose we have a principal-matching policy $\rho$ and an authorization policy $\pa$. If we require that the actions $\set{a_1,\dots,a_n}$ should each be performed by different users (and the same action may be repeated), we add the rules
\[
 (\audita{a_1},p_1), \dots, (\audita{a_n},p_n)
\]
to $\rho$ and let the new policy be $\rho'$. And for each principal $p_i$, we add the set of rules
\[
 \set{(p_i,o,a_j,0) :  1 \leq j \leq n, j \ne i}.
\]
to $\pa$ denoting the new policy $\pa'$. We then have the following result

\begin{proposition}\label{prop:separation-of-duty}
    Given an RPPM separation of duty policy, as described above, for any user $u$ the request $(u,o,a)$ is allowed if the request is authorized by $\rho'$ and $\pa'$ and no request of the form $(u,o,a')$, where $a' \neq a$, has been previously authorized; the request is denied otherwise.
\end{proposition}

\begin{proof}
    The proof proceeds by induction on the number of evaluated requests.
    Consider the (base) case when no requests have yet been made.
    A request $(u,o,a)$ where $a \in \set{a_1, \dots, a_n}$ will not match any of the $n$ inserted principal-matching rules as no decision audit edges currently exist in the system graph.
    Thus request $(u,o,a)$ will be authorized if it is authorized by $\rho$ and $\pa$ (and hence will be authorized by $\rho'$ and $\pa'$).

    Now suppose the result holds for all sequences of $m$ requests and consider the request $(u,o,a)$ where $a \in \set{a_1, \dots, a_n}$.
    \begin{itemize}
        \item If $u$ has previously performed a constrained action $a_i$, $1 \le i \le n$, then the request will satisfy principal-matching condition $(\audita{a_i}, p_i)$.
	
	      Now, if $a_i = a$, there is no authorization rule of the form $(p_i,o,a_i,0)$ and the request will, therefore, be authorized if and only if it is authorized by $\rho$ and $\pa$.
	
	      Conversely, if $a_i \ne a$, then $a = a_j$, for some $j \ne i$, and the authorization rule $(p_i, o, a, 0)$, together with the \textsf{DenyOverride} CRS will cause the request to be denied.
        \item If user $u$ has not previously performed a constrained action then the request will not match any of the principal-matching rules that were added to create $\rho'$.
	      Thus the request will only be authorized if it is authorized by $\rho$ and $\pa$.
    \end{itemize}
\end{proof}
\subsection{Enforcing Chinese Walls}\label{sec:audit:chinesewall}
The Chinese Wall principle may be used to control access to information in order to prevent any conflicts of interest arising.
The standard use case concerns a consultancy that provides services to multiple clients, some of whom are competitors.
It is important that a consultant does not access documents of company $A$ if she has previously accessed documents of a competitor of $A$.

To support the Chinese Wall policy, systems classify data using conflict of interest classes~\cite{BrewerN89}, indicating groups of competitor entities.
Requests to access a company's resources within a conflict of interest class will only be authorized if no previous request was authorized accessing resources from another company in that conflict of interest class.

Unlike the general approach for separation of duty, a general approach for Chinese Wall requires fewer policy changes but does rely on a particular basic layout of system graph.
This layout is such that the users who will be making requests are connected (directly or indirectly) to the companies (which may or may not be competitors of each other).
These companies are then connected to their respective data entities, which will be the targets of users' requests.
This arrangement is depicted, conceptually, in Figure~\ref{img:chinesewallconcept:b}, with the path condition $\pi_1$ representing the chain of relationships between users and companies and $\pi_2$ between the data entities and the companies.\footnote{It should be noted that Figure~\ref{img:chinesewallconcept} does not show system graphs; it shows high-level representations of the `shape' of a system graph.}
In other words, the path from an authorized user to a company will contain the same labels (and will match the path condition $\pi_1$), irrespective of the specific identities of the user and company.
Similarly, the path from a data object to its owner company will contain the same labels (and match the path condition $\pi_2$).
Thus, the principal that is authorized to access companies' data objects would be matched using the path condition $\pi_1 \comp \overline{\pi_2}$.

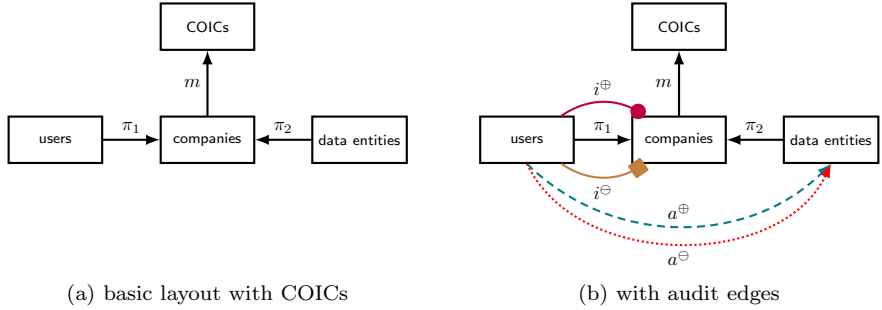
\begin{figure}[!ht]\centering
     \subfloat[basic layout with COICs]{
        \begin{tikzpicture}
            [node distance=1.5cm and 2cm, on grid,
            caching/.style={color=blue!70,>=open diamond}, 
            audit-a/.style={color=teal, densely dashed}, 
            audit-d/.style={color=red, densely dotted}, 
            interest-a/.style={color=purple,>=\interestahead}, 
            interest-b/.style={color=brown,>=\interestdhead}, 
            entity/.style={rectangle,draw,minimum width=50pt,minimum height=25pt},thick,
            every node/.append style={scale=0.7, transform shape}]
            \begin{scope}[>=latex] 
                \node[entity] (u) {\textsf{\footnotesize users}};
                \node[entity,right=of u] (c) {\textsf{\footnotesize companies}};
                \node[entity,right=of c] (d) {\textsf{\footnotesize data entities}};
                \node[entity,above=of c] (i) {\textsf{\footnotesize COICs}};
                \draw[thick,->] (u) to node[auto] {\textsf{$\pi_1$}} (c);
                \draw[thick,->] (d) to node[swap,auto] {\textsf{$\pi_2$}} (c);
                \draw[thick,->] (c) to node[auto] {\textsf{$m$}} (i);
                \draw[audit-a,thick,->,opacity=0] (u.south) to [bend right=45] node[swap,auto,color=black] {\textsf{\audita{a}}} (d.south);
                \draw[audit-d,thick,->,opacity=0] (u.south) to [bend right=65] node[swap,auto,color=black] {\textsf{\auditd{a}}} (d.south);
            \end{scope}
        \end{tikzpicture}
        \label{img:chinesewallconcept:b}
    }\qquad
     \subfloat[with audit edges]{
        \begin{tikzpicture}
            [node distance=1.5cm and 2cm, on grid,
            caching/.style={color=blue!70,>=open diamond}, 
            audit-a/.style={color=teal, densely dashed}, 
            audit-d/.style={color=red, densely dotted}, 
            interest-a/.style={color=purple,>=\interestahead}, 
            interest-b/.style={color=brown,>=\interestdhead}, 
            entity/.style={rectangle,draw,minimum width=50pt,minimum height=25pt},thick,
            every node/.append style={scale=0.7, transform shape}]
            \begin{scope}[>=latex] 
                \node[entity] (u) {\textsf{\footnotesize users}};
                \node[entity,right=of u] (c) {\textsf{\footnotesize companies}};
                \node[entity,right=of c] (d) {\textsf{\footnotesize data entities}};
                \node[entity,above=of c] (i) {\textsf{\footnotesize COICs}};
                \draw[thick,->] (u) to node[auto] {\textsf{$\pi_1$}} (c);
                \draw[thick,->] (d) to node[swap,auto] {\textsf{$\pi_2$}} (c);
                \draw[thick,->] (c) to node[auto] {\textsf{$m$}} (i);
                \draw[interest-a,thick,->] (u) to [bend left=35] node[auto,color=black] {\textsf{\interesta}} (c);
                \draw[interest-b,thick,->] (u) to [bend right=35] node[swap,auto,color=black] {\textsf{\interestb}} (c);
                \draw[audit-a,thick,->] (u.south) to [bend right=45] node[auto,color=black] {\textsf{\audita{a}}} (d.south);
                \draw[audit-d,thick,->] (u.south) to [bend right=65] node[swap,auto,color=black] {\textsf{\auditd{a}}} (d.south);
            \end{scope}
      \end{tikzpicture}
        \label{img:chinesewallconcept:c}
    }
    \caption{Chinese Wall Generalisation}\label{img:chinesewallconcept}
\end{figure}

To support the Chinese Wall constraint, the basic layout is supplemented by conflict of interest classes, to which companies are connected directly by the member ($m$) relationship (see Figure~\ref{img:chinesewallconcept:b}).
We assume here that membership of conflict of interest classes is determined when the system graph is initially populated and remains fixed through the lifetime of the system.
When users are authorized (or denied) access to particular data entities, authorized (or denied) decision audit edges will result for these requests as shown in Figure~\ref{img:chinesewallconcept:c}.
We additionally introduce interest audit edges into the system graph which are added between users and companies (see Figure~\ref{img:chinesewallconcept:c}).
Active interest audit edges are labelled with \interesta, blocked interest audit edges are labelled with \interestb.
We, therefore, extend the set of relationships to include the set $\setinterest$, thus allowing the system graph to support these new edges.
Graphically, we represent active interest audit edges with a filled circle head, whilst blocked interest audit edges have a filled square head.

Informally, when a subject's request to access a company's data is authorized, an active interest audit edge is added (if it doesn't already exist) between the subject and the company whose data was accessed.
(We will also add an authorized decision audit edge between the subject and the data entity if it does not already exist.)
Additionally, blocked interest audit edges are added (if they don't already exist) between the subject and all other companies who are members of the conflict of interest class to which the first company is a member.
Interest audit edges are not added after denied authorization requests.

For a concrete example, consider the system graph $G_4$ shown in Figure~\ref{img:system-graph-fragment4:a}, where a member of staff $u_1$ works for an employer $e_1$.
This employer supplies numerous clients ($c_1$, $c_2$ and $c_3$) which have data in the form of files ($f_1$, $f_2$, $f_3$ and $f_4$).
In this example users are connected to companies by $\pi_1 = w \comp s$ whilst data entities are connected to companies by $\pi_2 = d$.

\begin{figure}[!th]\centering
    \subfloat[system graph fragment]{
        \begin{tikzpicture}
            [node distance=1.5cm and 1.6cm, on grid,
            caching/.style={color=blue!70,>=open diamond}, 
            audit-a/.style={color=teal, densely dashed}, 
            audit-d/.style={color=red, densely dotted}, 
            interest-a/.style={color=purple,>=\interestahead}, 
            interest-b/.style={color=brown,>=\interestdhead}, 
            every circle node/.style={draw,minimum width=20pt},thick,
            every node/.append style={scale=0.7, transform shape}]
            \begin{scope}[>=latex] 
                \node[circle] (u1) {$u_1$};
                \node[circle,below=of u1] (e1) {$e_1$};
                \node[circle,left=of e1] (c1) {$c_1$};
                \node[circle,below=of e1] (c2) {$c_2$};
                \node[circle,right=of e1] (c3) {$c_3$};
                \node[circle,above=of c1] (f1) {$f_1$};
                \node[circle,left=of c1] (f4) {$f_4$};
                \node[circle,above=of c3] (f3) {$f_3$};
                \node[circle,below=of c2] (f2) {$f_2$};
                \node[circle,left=of c2] (i1) {$i_1$};
                \node[circle,right=of c2] (i2) {$i_2$};
                \draw[thick,->] (u1) to node[swap,auto] {\textsf{$w$}} (e1);
                \draw[thick,->] (e1) to node[swap,auto] {\textsf{$s$}} (c1);
                \draw[thick,->] (e1) to node[swap,auto] {\textsf{$s$}} (c2);
                \draw[thick,->] (e1) to node[auto] {\textsf{$s$}} (c3);
                \draw[thick,->] (f1) to node[swap,auto] {\textsf{$d$}} (c1);
                \draw[thick,->] (f2) to node[auto] {\textsf{$d$}} (c2);
                \draw[thick,->] (f3) to node[swap,auto] {\textsf{$d$}} (c3);
                \draw[thick,->] (f4) to node[auto] {\textsf{$d$}} (c1);
                \draw[thick,->] (c1) to node[swap,auto] {\textsf{$m$}} (i1);
                \draw[thick,->] (c2) to node[swap,auto] {\textsf{$m$}} (i1);
                \draw[thick,->] (c3) to node[swap,auto] {\textsf{$m$}} (i2);
            \end{scope}
        \end{tikzpicture}
        \label{img:system-graph-fragment4:a}
    }\qquad
     \subfloat[after request $q_4 = (u_1,f_3,\textsf{read})$]{
        \begin{tikzpicture}
            [node distance=1.5cm and 1.6cm, on grid,
            caching/.style={color=blue!70,>=open diamond}, 
            audit-a/.style={color=teal, densely dashed}, 
            audit-d/.style={color=red, densely dotted}, 
            interest-a/.style={color=purple,>=\interestahead}, 
            interest-b/.style={color=brown,>=\interestdhead}, 
            every circle node/.style={draw,minimum width=20pt},thick,
            every node/.append style={scale=0.7, transform shape}]
            \begin{scope}[>=latex] 
                \node[circle] (u1) {$u_1$};
                \node[circle,below=of u1] (e1) {$e_1$};
                \node[circle,left=of e1] (c1) {$c_1$};
                \node[circle,below=of e1] (c2) {$c_2$};
                \node[circle,right=of e1] (c3) {$c_3$};
                \node[circle,above=of c1] (f1) {$f_1$};
                \node[circle,left=of c1] (f4) {$f_4$};
                \node[circle,above=of c3] (f3) {$f_3$};
                \node[circle,below=of c2] (f2) {$f_2$};
                \node[circle,left=of c2] (i1) {$i_1$};
                \node[circle,right=of c2] (i2) {$i_2$};
                \draw[thick,->] (u1) to node[swap,auto] {\textsf{$w$}} (e1);
                \draw[thick,->] (e1) to node[swap,auto] {\textsf{$s$}} (c1);
                \draw[thick,->] (e1) to node[swap,auto] {\textsf{$s$}} (c2);
                \draw[thick,->] (e1) to node[auto] {\textsf{$s$}} (c3);
                \draw[thick,->] (f1) to node[swap,auto] {\textsf{$d$}} (c1);
                \draw[thick,->] (f2) to node[auto] {\textsf{$d$}} (c2);
                \draw[thick,->] (f3) to node[swap,auto] {\textsf{$d$}} (c3);
                \draw[thick,->] (f4) to node[auto] {\textsf{$d$}} (c1);
                \draw[thick,->] (c1) to node[swap,auto] {\textsf{$m$}} (i1);
                \draw[thick,->] (c2) to node[swap,auto] {\textsf{$m$}} (i1);
                \draw[thick,->] (c3) to node[swap,auto] {\textsf{$m$}} (i2);   	
                \draw[audit-a,thick,->] (u1) to node[swap,auto,color=black] {\textsf{\audita{\textsf{read}}}} (f1);
                \draw[interest-a,thick,->] (u1) to node[auto,color=black] {\textsf{\interesta}} (c1);
                \draw[interest-b,thick,->] (u1) to [bend right] node[swap,auto,pos=0.625,color=black] {\textsf{\interestb}} (c2);
                \draw[audit-d,thick,->] (u1) to [bend left] node[auto,near end,color=black] {\textsf{\auditd{\textsf{read}}}} (f2);
                \draw[audit-a,thick,->] (u1) to node[auto,color=black] {\textsf{\audita{\textsf{read}}}} (f3);
                \draw[audit-a,thick,->] (u1) to node[swap,auto,near end,color=black] {\textsf{\audita{\textsf{read}}}} (f4);
                \draw[interest-a,thick,->] (u1) to node[auto,color=black] {\textsf{\interesta}} (c3);
              \end{scope}
      \end{tikzpicture}
        \label{img:system-graph-fragment4:c}
    }
    \caption{Enforcing the Chinese Wall policy in RPPM}\label{img:system-graph-fragment4}
\end{figure}
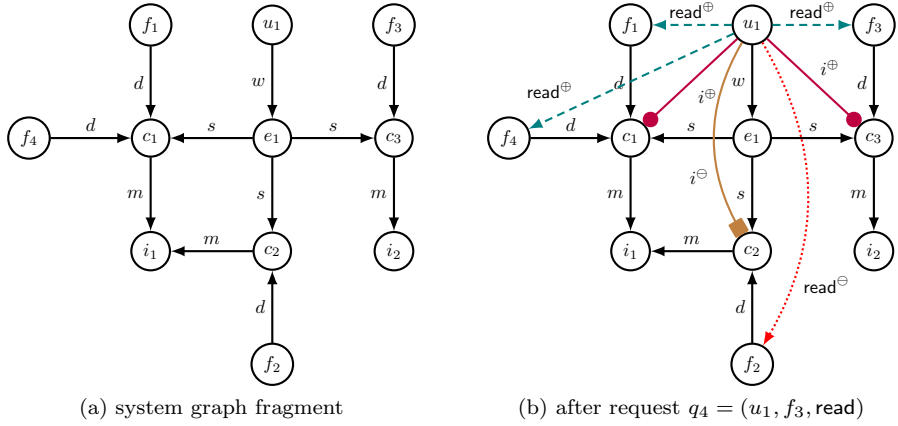

If we assume the existence of a principal-matching policy $\rho = [(w \comp s \comp \overline{d}, p)]$ and authorization policy $\pa = [(p,\star,\textsf{read},1)]$, then $u_1$ would be authorized to read all files.
However, the clients are members of conflict of interest classes ($i_1$ and $i_2$) with clients $c_1$ and $c_2$ being competitors in $i_1$.
Accordingly, we modify the principal-matching and authorization policies as follows:
\[
  \rho_{\it cw} = [(\interestb \comp \overline{d}, p_{\it cw}), (w \comp s \comp \overline{d}, p)]\quad\text{and}\quad \pa_{\it cw} = [(p_{\it cw},\star,\star,0),(p,\star,\textsf{read},1)].
\]

We now consider four different types of request that can arise.
Figure~\ref{img:system-graph-fragment4:c} shows the graph $G_4$ after all four requests have been made.
\begin{description}
    \item[Initial declaration of interest.]
    The request $q_1 = (u_1, f_1, \textsf{read})$ to read data belonging to client $c_1$ in the graph $G_4$ will be authorized:
    the first principal-matching rule is not matched but the second one is.
    Thus $\mp = [p]$ and the request is authorized, resulting in an authorized decision audit edge $(u_1, f_1, \audita{\textsf{read}})$ being added to the graph along with the interest edges $(u_1,c_1,\interesta)$ and $(u_1,c_2,\interestb)$. 

    \item[Continued interest.]
    If $u_1$ makes a second request $q_2 = (u_1,f_4,\textsf{read})$ for data of client $c_1$ this will also be authorized.
    The first principal-matching rule cannot be matched, as before, and the second can with principal $p$, once again, being authorized to read all objects.
    The authorized decision audit edge $(u_1, f_4, \audita{\textsf{read}})$ will be added to the graph but no new interest edges are added as the required edges already exist.

    \item[Conflict of interest request.]
    If $u_1$ requests data for a competing client $c_2$ using a third request $q_3 = (u_1,f_2,\textsf{read})$, this will be denied.
    This time, the principal-matching rule $(\interestb \comp \overline{d}, p_{\it cw})$ is matched and $p_{\it cw}$ is denied all actions on all objects.
    A denial audit edge $(u_1,f_2,\auditd{\textsf{read}})$ is added to the graph.

    \item[New declaration of interest which doesn't conflict.]
    Lastly $u_1$ makes a request $q_4 = (u_1,f_3,\textsf{read})$ for data of a third client $c_3$ who does not conflict with $c_1$ (with membership in a different conflict of interest class).
    As with the first two requests, this request will be authorized using the second principal-matching rule and matched principal $p$.
    An authorized decision audit edge $(u_1,f_3,\audita{\textsf{read}})$ is then added to the graph along with the active interest edge $(u_1,c_3,\interesta)$.
    No blocked interest edges are added as there are no companies other than $c_3$ who are members of conflict of interest class $i_2$.

\end{description}

More generally, suppose we have a principal-matching policy $\rho$.
In order to enforce the Chinese Wall constraint using this basic layout we add a new principal-matching rule $(\interestb \comp \overline{\pi_2}, p_{\it cw})$ to $\rho$ to produce a new policy $\rho_{\it cw}$.
The $p_{\it cw}$ principal is denied all actions on all data entities through the inclusion of an authorization rule $(p_{\it cw}, \star, \star, 0)$ into the existing authorization policy $\pa$, producing a new authorization policy $\pa_{\it cw}$.%
  \footnote{Once again, whilst we use the \textsf{AllMatch} PMS and the \textsf{DenyOverride} CRS we could equally employ the \textsf{FirstMatch} PMS with any CRS as long as we ensure that the constraint principal-matching rules are added before any existing rules.}

\begin{proposition}\label{prop:chinese-wall}
    Given an RPPM Chinese Wall constraint, as described above, for any user $u$ the request $(u,o,a)$ is allowed if the request is authorized by $\rho_{\it cw}$ and $\pa_{\it cw}$ and the user $u$ does not have an active interest in any company $c'$ which is a member of the same conflict of interest class as the company $c \neq c'$ responsible for $o$.
    In all other cases the request is denied.
\end{proposition}

\begin{proof}
    The proof proceeds by induction on the number of evaluated requests.
    Consider the (base) case when no requests have yet been made.
    A request $q_1 = (u', o, a_j)$ will not match the inserted principal-matching rule $(\interestb \comp \overline{\pi_2}, p_{\it cw})$ as no blocked interest audit edges currently exist in the system graph.
    By assumption, request $q_1$ will match a preexisting principal-matching rule with principal-matching condition $\pi_1 \comp \overline{\pi_2}$ and in doing so will match principal $p$.
    Also by assumption, principal $p$ is authorized to perform action $a_j$, therefore, the audit edge $(u',o,\audita{a_j})$ will be added to the system graph.
    Additionally, the active interest edge $(u',c,\interesta)$ will be added where $c$ represents the company the target data entity $o$ belongs to (i.e. there is a path of relations satisfying $\pi_2$ between $o$ and $c$ as required by the basic layout).
    Lastly, blocked interest edges $(u',c',\interestb)$ will be added for each company $c' \neq c$ in the same conflict of interest class; these companies are identified through the existence of edges $(c',i',m)$ in the system graph where there is also an edge $(c,i,m)$ with $i' = i$.

    Now consider the case when request $q_{x+1} = (u'',o',a_k)$ is made after $x$ requests have been successfully evaluated.
    We assume, without loss of generality, that data entity $o'$ belongs to company $c_1$, a member of conflict of interest class $i_1$.
    \begin{itemize}
        \item If user $u''$ has no active interests in any company, then request $q_{x+1}$ will not match the inserted principal-matching rule $(\interestb \comp \overline{\pi_2}, p_{\it cw})$ as no blocked interest audit edges currently exist in the system graph for $u''$.
            By assumption, request $q_1$ will match a preexisting principal-matching rule with principal-matching condition $\pi_1 \comp \overline{\pi_2}$ and in doing so will match principal $p$.
            Also by assumption, principal $p$ is authorized to perform action $a_k$, therefore, the audit edge $(u',o',\audita{a_k})$ will be added to the system graph.
            The active interest edge $(u'',c_1,\interesta)$ will be added to the system graph, as will blocked interest edges $(u'',c_y,\interestb)$ for each company $c_y \neq c_1$ who is a member of the conflict of interest class $i_1$.

        \item If user $u''$ has an active interest in company $c_1$, then request $q_{x+1}$ will not match the inserted principal-matching rule $(\interestb \comp \overline{\pi_2}, p_{\it cw})$ as an active, rather than blocked, interest audit edge exists between $u''$ and $c_1$.
            By assumption, request $q_1$ will match a preexisting principal-matching rule with principal-matching condition $\pi_1 \comp \overline{\pi_2}$ and in doing so will match principal $p$.
            Also by assumption, principal $p$ is authorized to perform action $a_k$, therefore, the audit edge $(u',o',\audita{a_k})$ will be added to the system graph.
            The active interest edge $(u'',c_1,\interesta)$ will not be added to the system graph as it already exists.
            Blocked interest edges $(u'',c_y,\interestb)$ for each company $c_y \neq c_1$ who is a member of the conflict of interest class $i_1$ will be added where they do not already exist.

        \item If user $u''$ has an active interest in company $c_2$ which is a member of the same conflict of interest class $i_1$ as $c_1$, then request $q_{x+1}$ will, in this instance, match the inserted principal-matching rule $(\interestb \comp \overline{\pi_2}, p_{\it cw})$.
            As the principal $p_{\it cw}$ applies the inserted authorization rule $(p_{\it cw},\star,\star,0)$ overrides the assumed authorization achieved through principal $p$.
            The denied decision audit edge $(u',o',\auditd{a_k})$ will be added to the system graph and no interest audit edges will be added.

        \item If user $u''$ has an active interest in company $c_3$ which is a member of a different conflict of interest class $i_2$ to $c_1$ but no active interest in any company which is a member of the conflict of interest class $i_1$ to which $c_1$ is a member, then request $q_{x+1}$ will not match the inserted principal-matching rule $(\interestb \comp \overline{\pi_2}, p_{\it cw})$ as no interest audit edge exists between $u''$ and $c_1$.
            By assumption, request $q_1$ will match a preexisting principal-matching rule with principal-matching condition $\pi_1 \comp \overline{\pi_2}$ and in doing so will match principal $p$.
            Also by assumption, principal $p$ is authorized to perform action $a_k$, therefore, the audit edge $(u',o',\audita{a_k})$ will be added to the system graph.
            The active interest edge $(u'',c_1,\interesta)$ will be added to the system graph, as will blocked interest edges $(u'',c_y,\interestb)$ for each company $c_y \neq c_1$ who is a member of the conflict of interest class $i_1$.
    \end{itemize}
\end{proof}

The basic model described above is consistent with that used by Brewer and Nash, where there is a simple and fixed relationship between users and companies (path condition $\pi_1$) and between data objects and companies (path condition $\pi_2$).
However, this approach is unnecessarily restrictive (and was chosen for ease of exposition), in that we may wish to define more complex authorization requirements between such entities.
In practice, there is no reason why multiple path conditions cannot be used between users, objects and companies, each of which is mapped to the appropriate principal.

For example, given two paths of relations between users and companies ($w \comp s$ and $w \comp p \comp s$) and two paths of relations between data entities and companies ($d$ and $f \comp d$) the principal-matching policy from our running example is modified to include both options for blocking paths and all combinations for normal authorization.
\begin{align*}
    \rho_{{\it cw}_2} = [&(\interestb \comp \overline{d}, p_{{\it cw}_2}), (\interestb \comp \overline{d} \comp \overline{f}, p_{{\it cw}_2}), \\
                                &(w \comp s \comp \overline{d}, p), (w \comp s \comp \overline{d} \comp \overline{f}, p), (w \comp p \comp s \comp \overline{d}, p), (w \comp p \comp s \comp \overline{d} \comp \overline{f}, p)]
\end{align*}


\section{Related Work}\label{sec:relatedwork}

Relationship-based access control is becoming an increasingly important alternative approach to specifying and enforcing authorization policies.
A number of models have been proposed in recent years~\cite{CarminatiFP09,ChengPS12passat,ChengPS12dbsec,CramptonS14,Fong11,ZhangAGC09}, but most have focused on access control in social networks~\cite{CarminatiFP09,ChengPS12passat,ChengPS12dbsec,Fong11,ZhangAGC09}.
In this paper, we extend the RPPM model of Crampton and Sellwood~\cite{CramptonS14} by introducing additional types of edges to support efficient request evaluation and history-based access control policies.

History-based access control, where an authorization decision is dependent (in part) on the outcome of previous requests, has been widely studied since Brewer and Nash's seminal paper on the Chinese Wall policy~\cite{BrewerN89}.
The enforcement mechanism for this policy is based on a history matrix, which records what requests have previously been allowed.
It is very natural to record such information as audit edges in the system graph and to use these edges to define and enforce history-based policies.
Fong {\em et al.} recently proposed a relationship-based model that incorporated temporal operators, enabling them to specify and enforce history-based policies~\cite{FoMeKr13}.
This work extended Fong's ReBAC model, developed in the context of social networks, and is thus unsuitable for the more generic access control applications for which the RPPM model was designed.
In particular, there is no obvious way in which it can support Chinese wall policies.

There has been some interest in recent years in reusing, recycling or caching authorization decisions at policy enforcement points in order to avoid recomputing decisions~\cite{BordersZP05,KohlerBS09,KohlerF09,WeiCB11}.
In principle, these techniques are particularly valuable in large-scale, distributed systems, providing faster decisions, and the potential to recover from failures in the communication infrastructure or failure of one of the components, in the (distributed) authorization system.
However, many of the techniques are of limited value because the correlation between access control decisions and the structure of access control policies is typically rather low.
In contrast, caching in the RPPM model has the potential to substantially speed up decision-making because a cached edge is of real value as it enables the decision-making apparatus to sidestep the expensive step of principal matching and proceed directly to evaluating the authorization policy.
Moreover, a cached edge applies to multiple requests, irrespective of whether the request has previously been evaluated, unlike many, if not all, proposals in the literature.

\section{Conclusion}\label{sec:conclusion}

The RPPM model fuses ideas from relationship-based access control (by using a labelled graph), role-based access control (in its use of principals to simplify policy specification) and Unix (by mapping a user-object relationship to a principal before determining whether a request is authorized).
This unique blend of features make it suitable for large-scale applications in which the relationships between users are a crucial factor in specifying authorization rules.

In addition to these advantages, the RPPM model is particularly suitable for recording information that may be generated in the process of making authorization decisions.
In this paper, we focus on two new types of edges.
Caching edges introduce shortcuts in the system graph indicating the principals associated with a user-object pair.
Such edges can introduce substantial efficiencies to the evaluation of decisions.
Audit edges allow for the enforcement of history-based policies, including separation of duty and Chinese wall policies.

The introduction of audit edges lays the foundation for future work supporting workflow tasks within the RPPM model.
This work may, additionally, require the model to be further extended with the introduction of stateful entities.

\bibliography{../Relix2AC}
\bibliographystyle{abbrv}

\end{document}